%% file: main.tex
\documentclass[11pt]{article}

\input{macros.tex}
\title{Competitive Algorithms for Block-Aware Caching}

\author{Christian Coester\thanks{University of Sheffield, Sheffield, UK. Email:
		\texttt{christian.coester@gmail.com}. Part of this research was done while the author was at Tel Aviv University, supported by the Israel Academy of Sciences and Humanities \& Council for Higher Education Excellence Fellowship Program for International Postdoctoral Researchers.}
	\and
	Roie Levin\thanks{Computer Science Department, Carnegie Mellon
				University, Pittsburgh, PA 15213. Email:
				\texttt{roiel@cs.cmu.edu}. Part of this research was done while the author was at the Technion. Supported in part by US-Israel BSF grant 2018352 and by ISF grant
				 2233/19 (2027511).}
	\and
	Joseph (Seffi) Naor\thanks{Computer Science Department,
		Technion, Haifa, Israel. Emails: \texttt{\{naor,ohad\}@cs.technion.ac.il}. Supported in part by US-Israel BSF grant 2018352 and by ISF grant
		 2233/19 (2027511).}
	\and
	Ohad Talmon\footnotemark[3]{} 
}

\date{}

\begin{document}
	
	\maketitle

\input{abstract.tex}

	\input{intro.tex}
	\input{eviction.tex}

	\input{loading.tex}
	
	{\footnotesize
		\bibliography{refs}
		\bibliographystyle{alpha}
	}

	\input{appendix.tex}

\end{document}

%% file: macros.tex
\usepackage{typearea}
\typearea{14}
\usepackage{setspace}
\usepackage[T1]{fontenc}

\usepackage{bbm}
\usepackage{framed} 
\usepackage{url} 
\usepackage{booktabs}
\usepackage{amsmath,amssymb}
\usepackage{amsthm}
\usepackage{thmtools} 
\usepackage{thm-restate}
\usepackage{nicefrac}
\usepackage{microtype}
\usepackage{calc}
\usepackage{enumerate}
\usepackage{enumitem}
\usepackage[usenames,dvipsnames]{xcolor}
\usepackage[colorlinks,citecolor=blue,linkcolor=BrickRed]{hyperref}
\usepackage{algorithm}
\usepackage{graphicx}
\usepackage[font=footnotesize,labelfont=bf]{subcaption}
\usepackage[font=footnotesize,labelfont=bf]{caption}
\usepackage[nobreak=true]{mdframed}
\usepackage{appendix}
\usepackage[noend]{algpseudocode}
\usepackage[colorinlistoftodos]{todonotes}
\usepackage{xr}
\usepackage{array,longtable}
\usepackage{xspace}

\usepackage[capitalise]{cleveref}
\crefrangelabelformat{enumi}{#3#1#4--#5#2#6}
\usepackage{chngcntr}
\usepackage{mathtools}

\usepackage{nameref}

\usepackage{tikz}
\usetikzlibrary{patterns}

\allowdisplaybreaks

\newcommand\bsize{\beta}

\DeclareMathOperator*{\expectation}{\mathbb{E}}

\newcommand\unvrs{\mathcal{N}}

\newcommand\opt{\textsc{Opt}\xspace}

\counterwithin{equation}{section}

\usepackage{eqparbox}

\newcommand\cmin{c_{\min}}
\newcommand\cmax{c_{\max}}

\newcommand\subcov{\textsc{SubmodularCover}\xspace}

\theoremstyle{plain}
\newtheorem{theorem}{Theorem}[section]
\newtheorem{definition}[theorem]{Definition}

\newtheorem{lemma}[theorem]{Lemma}

\newtheorem{claim}[theorem]{Claim}

\newtheorem{corollary}[theorem]{Corollary}

\newlength{\continueindent}
\usepackage{etoolbox}
\makeatletter
\newcommand*{\ALG@customparshape}{\parshape 2 \leftmargin \linewidth \dimexpr\ALG@tlm+\continueindent\relax \dimexpr\linewidth+\leftmargin-\ALG@tlm-\continueindent\relax}
\apptocmd{\ALG@beginblock}{\ALG@customparshape}{}{\errmessage{failed to patch}}
\makeatother

\makeatletter
\def\thm@space@setup{%
	\thm@preskip=\parskip \thm@postskip=0pt
}
\makeatother

\usepackage{etoolbox}
\usepackage{tikz}
\usetikzlibrary{tikzmark}
\usetikzlibrary{calc}

\errorcontextlines\maxdimen

\newcommand{\ALGtikzmarkcolor}{black}
\newcommand{\ALGtikzmarkextraindent}{4pt}
\newcommand{\ALGtikzmarkverticaloffsetstart}{-.5ex}
\newcommand{\ALGtikzmarkverticaloffsetend}{-.5ex}
\makeatletter
\newcounter{ALG@tikzmark@tempcnta}

\newcommand\ALG@tikzmark@start{%
	\global\let\ALG@tikzmark@last\ALG@tikzmark@starttext%
	\expandafter\edef\csname ALG@tikzmark@\theALG@nested\endcsname{\theALG@tikzmark@tempcnta}%
	\tikzmark{ALG@tikzmark@start@\csname ALG@tikzmark@\theALG@nested\endcsname}%
	\addtocounter{ALG@tikzmark@tempcnta}{1}%
}

\def\ALG@tikzmark@starttext{start}
\newcommand\ALG@tikzmark@end{%
	\ifx\ALG@tikzmark@last\ALG@tikzmark@starttext
	\else
	\tikzmark{ALG@tikzmark@end@\csname ALG@tikzmark@\theALG@nested\endcsname}%
	\tikz[overlay,remember picture] \draw[\ALGtikzmarkcolor] let \p{S}=($(pic cs:ALG@tikzmark@start@\csname ALG@tikzmark@\theALG@nested\endcsname)+(\ALGtikzmarkextraindent,\ALGtikzmarkverticaloffsetstart)$), \p{E}=($(pic cs:ALG@tikzmark@end@\csname ALG@tikzmark@\theALG@nested\endcsname)+(\ALGtikzmarkextraindent,\ALGtikzmarkverticaloffsetend)$) in (\x{S},\y{S})--(\x{S},\y{E});%
	\fi
	\gdef\ALG@tikzmark@last{end}%
}

\apptocmd{\ALG@beginblock}{\ALG@tikzmark@start}{}{\errmessage{failed to patch}}
\pretocmd{\ALG@endblock}{\ALG@tikzmark@end}{}{\errmessage{failed to patch}}
\makeatother

\algblock[with]{With}{EndWith}
\algblockdefx[With]{With}{EndWith}%
[1]{\textbf{with} #1 \textbf{do}}%
{}

\makeatletter
\ifthenelse{\equal{\ALG@noend}{t}}%
{\algtext*{EndWith}}
{}%
\makeatother

\newcommand{\expect}{\expectation\expectarg}
\DeclarePairedDelimiterX{\expectarg}[1]{[}{]}{%
	\ifnum\currentgrouptype=16 \else\begingroup\fi
	\activatebar#1
	\ifnum\currentgrouptype=16 \else\endgroup\fi
}

\newcommand{\expectover}[1]{\expectation_{#1}\expectarg}

\newcommand{\innermid}{\nonscript\;\delimsize\vert\nonscript\;}
\newcommand{\activatebar}{%
	\begingroup\lccode`\~=`\|
	\lowercase{\endgroup\let~}\innermid 
	\mathcode`|=\string"8000
}
\usepackage{parskip}

%% file: abstract.tex
\begin{abstract}
	Motivated by the design of real system storage hierarchies, we study the \emph{block-aware caching problem}, a generalization of classic caching in which fetching (or evicting) pages from the same block incurs the same cost as fetching (or evicting) just one page from the block. Given a cache of size $k$, and a sequence of requests from $n$ pages partitioned into given blocks of size $\bsize\leq k$, the goal is to minimize the total cost of fetching to (or evicting from) cache. This problem captures generalized caching as a special case, which is already NP-hard offline. We show the following suite of results:
	
	\begin{itemize}
		\item For the eviction cost model, we show an $O(\log k)$-approximate offline algorithm, a $k$-competitive deterministic online algorithm, and an $O(\log^2 k)$-competitive randomized online algorithm. 
		\item For the fetching cost model, we show an integrality gap of $\Omega(\bsize)$ for the natural LP relaxation of the problem, and an $\Omega(\bsize + \log k)$ lower bound for randomized online algorithms. The strategy of ignoring the block-structure and running a classical paging algorithm trivially achieves an $O(\bsize)$ approximation and an $O(\bsize \log k)$ competitive ratio respectively for the offline and online-randomized setting. 
		\item For both fetching and eviction models, we show improved bounds for the $(h,k)$-bicriteria version of the problem. In particular, when $k=2h$, we match the performance of classical caching algorithms up to constant factors.
	\end{itemize}
	Our results establish a strong separation between the tractability of the fetching and eviction cost models, which is interesting since fetching/eviction costs are the same up to an additive term for the classic caching problem. Previous work of Beckmann et al. (SPAA 21) only studied online deterministic algorithms for the fetching cost model when $k > h$. 
	
	Our insight is to relax the block-aware caching problem to a submodular covering linear program. The main technical challenge is to maintain a competitive fractional solution to this LP, and to round it with bounded loss, as the constraints of this LP are revealed online. We hope that this framework is useful going forward for other problems that can be captured as submodular cover.
\end{abstract}

%% file: intro.tex
\section{Introduction}

Caching (also known as paging) has been extensively studied since the early days of
online computation and competitive analysis, establishing itself as a cornerstone problem in this field, see e.g.,
~\cite{ST85,F+,Y91,Y98,GS,BBN07,BBN08,ACN,AAK99,BBK,BFT96,I98,I97,I02}. Recent years have witnessed increased activity on non-standard caching models, e.g., elastic caches~\cite{GuptaK0P19}, caching with time windows~\cite{GKP20}, caching
with dynamic weights \cite{EvenMR18}, caching with machine learning predictions~\cite{LykourisV18}, and writeback-aware caching \cite{beckmann2020writeback, bansal2021efficient}. Many of the recent developments in competitive analysis, e.g., the
online primal-dual method, projections, and mirror descent~\cite{BNsurvey, BCN14, BCLLM18} are all rooted in online paging. We study here {\em block-aware caching}, a non-standard caching model studied recently, as well as its generalizations.

In the (classic) weighted paging problem there is a universe of $n$ pages, a cache that can hold up to $k$ pages, and each page is associated with a weight (fetch cost). At each time step a page is requested; if the requested page
is already in cache then no cost is incurred, otherwise
the page must be fetched into the cache, incurring a cost equal to its weight.
The goal is to minimize the total cost incurred. This problem is well studied and understood, and we briefly mention the main results known for it. 

Sleator and Tarjan \cite{ST85}, in their seminal paper on competitive analysis,
showed that any deterministic algorithm is at least $k$-competitive, and that LRU (Least Recently Used) is precisely $k$-competitive
for unweighted paging (i.e., all weights are equal). 
The $k$-competitive bound was later generalized to weighted paging as well \cite{CL91,Y94}.
When randomization is allowed, Fiat et al ~\cite{F+} gave the  
elegant Randomized Marking algorithm for unweighted paging, which is $\Theta(\log k)$-competitive against an oblivious adversary. 
For weighted paging, Bansal et al.~\cite{BBN07} gave an $O(\log k)$-competitive randomized algorithm 
using the online primal-dual framework~\cite{BNsurvey, AAABN03}. It uses a two-step approach. First, a deterministic
competitive algorithm is designed for a fractional version of the problem.
Then, a randomized online algorithm is obtained
by {\em rounding} the deterministic fractional solution online.

\paragraph{Block-aware caching.}
Real storage systems operate by constructing a hierarchy of memory levels, starting from a very fast and small memory (e.g., an SRAM cache) to a very large and slow memory (e.g., flash or disk).
The data items in each level are typically organized in {\em blocks}, and fetching (or evicting) data items from the same block incurs the same cost as fetching (or evicting) just a single item from the block. Using fetching costs models scenarios in which data is read-only; using eviction costs models scenarios in which data must be written to slow memory upon eviction, and the writing cost dominates the reading cost (see e.g. \cite{BGHM20,bansal2021efficient}). 

Thus, a natural question is how can one optimize cache performance by taking advantage of granularity changes across different storage hierarchy levels. 
This question was recently raised by Beckmann et al. \cite{beckmann2021brief}, who defined the {\em block-aware caching problem}, generalizing the classic paging problem, as follows. 
Given a cache of size $k$, and a sequence of requests from $n$ pages that are partitioned into given blocks of size $\bsize\leq k$, minimize the total cost of fetching (or evicting) from the cache so as to serve the requests\footnote{We note that Beckmann et al. \cite{beckmann2021brief}  considered block-aware caching only in the fetching cost model.}.

Block-aware caching also arises in web and cloud settings, where data items can be aggregated into chunks (i.e., blocks) of data, such that accessing a whole chunk incurs the same cost as accessing just a single item. Consider a distributed cluster of servers, where a common cache of data items is maintained. One such example is the ZFS distributed file system that
aggregates different devices into a single storage pool acting as an arbitrary data store. 
When accessing a server for a specific data item, the main cost paid (e.g., latency) is for accessing the server. The notion of a block of data in this setting corresponds to the largest chunk of data items that can be fetched from (or evicted to) a server, while maintaining that the cost of this operation is dominated by the cost of accessing the server. 
Web caching is another example of block-aware caching. Consider a content delivery network (CDN) that maintains a cache of data items and suppose the CDN connects to a website so as to access a data item (see, e.g., \cite{HasanGDS14,SongBLL20}). Typically, TCP/IP  provides a time window for connecting to the website and accessing the data item. Hence, it might be beneficial to fetch (or evict) many data items that belong to the website, and not just the particular data item that is currently accessed. Thus, the notion of a block of data in this setting corresponds to the maximum number of such data items that can be sent without increasing the travel time. 

In the generalized caching problem \cite{BBN08,ACER19}, pages are associated with both a size and a cost. At any point of time, the sum of the sizes of the pages in the cache cannot exceed the cache size. In the offline setting, generalized caching is known to be NP-hard, and in the online setting the known competitive factors for generalized caching \cite{BBN08,ACER19} match those of weighted caching.
It is not hard to see that block-aware caching captures generalized caching as a special case. Replace a page $p$ of size $s$ by a block $B$ of size $s$ containing page $p$ partitioned into unit size ``slices". The cost of accessing each slice is equal to the cost of $p$. Now, a request to page $p$ is replaced by many requests to the slices in $B$. Thus, an optimal solution to the block-aware caching problem generated has to fetch the full block $B$ into the cache.

\paragraph{Eviction and fetching costs.}
In classic paging, costs can be associated with either evicting or fetching pages. Clearly, for a given request sequence, optimal eviction and fetching costs of serving the requests can differ by at most an additive constant that only depends on the initial contents of the cache. However, this is not the case for block-aware caching, as optimal eviction and fetching costs can differ significantly, separating the two cost models. (We provide an example in Section \ref{prelim}.) As discussed above, the two cost models are practically motivated for block-aware caching, and we thus study both of them in this paper. We note that in the eviction cost model we are able to circumvent known lower bounds \cite{beckmann2021brief} that hold in the fetching cost model.

\subsection{Results and Techniques}

Observe that if an algorithm is $r$-competitive for classical paging, then it is at most $\bsize \cdot r$-competitive for block-caching in both fetching/eviction cost models; the reason is simply that \opt can be simulated by a classical paging algorithm that performs any single batched fetch/eviction in at most $\bsize$ rounds. With this in mind, our goal in this work is to beat this trivial linear dependence on $\bsize$.

Indeed, for the eviction cost model, we give the first set of algorithms avoiding a trivial multiplicative $\beta$ overhead over their classical paging counterparts. We also give $(h,k)$-bicriteria\footnote{In $(h,k)$ paging, an online algorithm with cache size $k$ competes against an offline cache of size $h$, where $k>h$.} algorithms for both fetching and eviction cost models, which we in turn use to adapt the lower bound of \cite{beckmann2021brief} for the fetching cost model to randomized algorithms.

\paragraph{Eviction cost.} We start in \cref{sec:eviction} with our main contributions: competitive algorithms for the eviction cost model. 
We show the following theorem.
\begin{theorem}
	For the block-aware caching problem with eviction cost, there exist:
	\begin{itemize}
		\item a $k$-competitive deterministic online algorithm.
		\item an $O(\log^2 k)$-competitive randomized (integral) online algorithm.
		\item an $O(\log k)$-approximate randomized offline algorithm.
	\end{itemize}
\end{theorem}
In fact, we study a more general version than the one introduced in \cite{beckmann2021brief} in which every block $B$ may have a separate cost $c_B$. For this more general weighted setting we get competitive ratios of $k$, $O(\log k \log (k \Delta))$ and $O(\log (k\Delta))$  for the deterministic online, randomized online, and randomized offline settings respectively (where $\Delta$ is the aspect ratio, i.e., the maximum cost ratio between any two blocks).

A first main technical ingredient is a linear programming relaxation for block-caching. It is tempting to use a formulation with a variable $x_p^t$ for each page $p$ and time $t$ indicating whether $p$ is present in cache in step $t$. However, in this case the eviction cost becomes a complicated non-linear function of the $x_p^t$. Instead, we define variable $\phi_B^t$ for each block $B$ and time step $t$ indicating whether we evict $B$ at $t$. This is reminiscent of the linear program for classical paging of \cite{BBN07} in which every variable represents whether a page is present in cache between two subsequent requests to a page, only that it may now be necessary to evict pages at any point between subsequent requests. 

A na\"{\i}ve linear programming formulation has an integrality gap of $\bsize$ (see \cref{sec:basic_LP_formulation}): this is unsurprising since the na\"{\i}ve LP exhibits this gap even for the special case of generalized paging. To get around this, we express feasibility as the constraint that a particular sequence of monotone, submodular functions is maximized. We then make use of good (albeit exponential size) LP relaxations for these submodular set function constraints, which were discovered by Wolsey \cite{Wolsey1982}. Our formulation may be viewed as a generalization of the strengthened LP relaxation due to \cite{BBN08} for generalized caching, which used the so-called knapsack cover (KC) inequalities. We first use the relaxation to give a $k$-competitive deterministic online algorithmic in \cref{sec:deterministic}. 

Next, we develop an $O(\log k)$-competitive fractional algorithm in \cref{sec:fractional}, followed by an $O(\log \Delta k)$-competitive online randomized rounding procedure in \cref{sec:rounding}. Combined, these results imply our $O(\log k \log \Delta k)$-competitive randomized online algorithm. It is natural to try to adapt the continuous online primal-dual framework of \cite{BBN07,BBN08}; however, owing to the increased complexity of our LP, there are several technical roadblocks. For one, our formulation now has primal variables corresponding to the eviction of every block at every point in time, and a na\"{\i}ve  adaptation of the continuous dynamics of \cite{BBN08} incurs loss that depends on the length of the request sequence. Nevertheless, we show how to carefully set the rate of increase of primal variables (with respect to the dual rate of increase) to construct a feasible solution with our claimed guarantee. For convenience, we do not present the fractional algorithm as online in the strict sense as we allow it to change decisions made in the past. However, it has the crucial property that it only increases LP variables, which suffices for the online rounding procedure.

For the rounding step, we forego maintaining an explicit distribution over cache states as in previous work \cite{BBN07,BBN08,ACER19}, since it is unclear how to control the cost of the rebalancing stage when updating the distribution in each time step. Instead, we use the method of random rounding with alterations in a similar spirit to \cite{bansal2021efficient}. A key difference in our work is the added difficulty of working with the submodular cover formulation: this introduces additional challenges to the analysis of the rounding (as well as to the maintenance of the fractional solution). We make use of recent work on online submodular cover \cite{gupta2020online}; interestingly we are able to charge our alteration cost to the fractional \emph{fetching} cost, even though this may be a factor $\bsize$ larger than the eviction cost.

\paragraph{Fetching cost.} We turn in \cref{sec:loading} to the fetching cost model,
where we show strong lower bounds, implying that the integrality gap of $\Omega(\bsize)$ of the natural LP formulation cannot be circumvented. We prove the following theorem for $(h,k)$ block-aware caching\footnote{Beckmann et al. \cite{beckmann2021brief} showed several deterministic lower bounds on the competitive factor achievable for $(h,k)$ block-aware caching, when $k\geq h+\beta$-1.}. 
\begin{theorem}
	When $k=O(h)$, no randomized online algorithm has competitive ratio better than $\Omega(\bsize + \log k)$ for block aware caching with fetching costs.
\end{theorem}
Our main idea here is an online deterministic rounding procedure for fractional algorithms that incurs constant blowup in both cache usage and cost. This implies an online derandomization procedure for any randomized algorithm, which in turn strengthens the lower bounds for deterministic algorithms of \cite{beckmann2021brief} to apply to randomized algorithms as well. Our lower bound implies that beating the trivial linear dependence on $\bsize$ is \emph{not} possible for the fetching cost model.

Our deterministic rounding procedure immediately implies improved bounds for the offline $(h,k)$ block-aware caching problem. In particular, when $k=2h$, we match the performance of classical caching algorithms up to constant factors.

\section{Model and Preliminaries}
\label{prelim}

\subsection{Problem Definition}
In the block-aware caching problem, there is a cache of size $k$ and $n$ pages which are partitioned into blocks. Let $\mathcal{B}$ be the partition of the pages into blocks. Each block contains at most $\bsize$ pages, for some $\bsize\in[k]$. For a block $B \in \mathcal{B}$, we denote by $c_B>0$ its cost. At each time-step $t$, a page $p_t$ is requested. To serve the request, the page $p_t$ must be fetched into the cache if it is missing from the cache. The goal is to obtain a feasible cache policy while minimizing the total cost. We consider two different cost functions.

\paragraph{Eviction cost model.}
In this model, fetching into the cache is free, while evictions have a cost that can be aggregated: Evicting any subset $A$ of a block $B$ at a time-step has a cost of $c_B$. The goal is to minimize the total eviction cost.

\paragraph{Fetching cost model.}
In this model, evicting pages from the cache is free, while fetching of pages has a cost that can be aggregated. Fetching of any subset $A$ of a block $B$ at a time-step has a cost of $c_B$. The goal is to minimize the total fetching cost.

Unlike classic paging and its variants, the fetching cost and eviction cost models are not equivalent in block-aware caching.  We show that the optimal fetching and eviction costs for the same request sequence may be off by a factor of $\bsize$ (in either direction!), and this bound is tight. 

\begin{restatable}{claim}{betaoff}
	\label{claim:beta_off}
	There exist instances of block-aware caching for which the optimal fetching cost is $\bsize$ larger than the optimal eviction cost, and there exist instances where the optimal eviction cost is $\bsize$ larger than the optimal fetching cost.
\end{restatable}
See \cref{sec:extra_proofs} for the proof. 

For a page $p$, define $B(p)$ to be the block containing $p$, and let $r(p,t)$ be the time of the last request to $p$ up until (and including) time $t$; if there is no such request, then $r(p,t):=-\infty$. Define the aspect ratio $\Delta :=  \cmax/ \cmin$ where $\cmax := \max_{B \in \mathcal{B}} c(B)$ and $\cmin := \min_{B \in \mathcal{B}} c(B)$. For convenience we will use the notation $[\ell] = \{1, \ldots, \ell\}$ and $[\ell]_0 = \{0, 1, \ldots, \ell\}$.

Our work relies on the theory of submodular functions which we introduce now for completeness.

\paragraph{Submodularity.} 
We consider set functions of the form $f: 2^{\unvrs} \rightarrow \mathbb{R}^+$, where $\unvrs$ is a set. For $A \subseteq B \subseteq \unvrs$, let $f(A \mid B) := f(A \cup B) - f(B)$. For convenience, if $A=\{v\}$ is a singleton we also write $f(v\mid B):=f(\{v\}\mid B)$. We call $f$ \textit{submodular} if for all $v\in\unvrs$, $A\subseteq B\subseteq\unvrs$ we have $f(v\mid A)\ge f(v\mid B)$. A simple result is that if a set function $f$ is submodular, then $f(\,\cdot\mid B)$ is also submodular for any
$B \subseteq \unvrs$. 
If for all $A \subseteq B \subseteq \unvrs$ we have that $f(A) \leq f(B)$, then we say that $f$ is \textit{monotone}.

\paragraph{Submodular Cover.}
Wolsey \cite{Wolsey1982} introduced the following problem known as \emph{submodular cover}. Given a monotone, submodular function $f$ over ground set $\mathcal{N}$, and cost function $c: \mathcal{N} \rightarrow \mathbbm{R}^+$ on the ground set, output a minimum cost subset $S \subseteq \mathcal{N}$ such that $f(S) \geq f(\mathcal{N})$. Wolsey gave the following LP relaxation:
\begin{align}
	\label{eq:wolsey_lp}
	\begin{array}{|rl|}
		\hline & \\
		\min & \displaystyle \sum_{v \in \mathcal{N}} c(v) \cdot x_v \\
		\text{subject to} & \\ & \\
		\forall S\subseteq \mathcal{N}: & \displaystyle \sum_{v \not \in S} f(v \mid S) \cdot x_v \geq f(\mathcal{N}) - f(S) \\ 
		\forall v \in \mathcal{N}: & x_v \geq 0  \\  & \\ \hline
	\end{array}
\end{align}
The constraints of this LP may be viewed as knapsack cover inequalities for a linearized version of the function $f$. Wolsey proved:
\begin{claim}[Proposition 2 of \cite{Wolsey1982}]
	\label{lem:wolsey_integerpts}
	A set $S$ has $f(S) = f(\mathcal{N})$ if and only if $\chi_S$, the
	characteristic vector of $S$, is a feasible integer solution to
	\eqref{eq:wolsey_lp}.
\end{claim}
Furthermore, Wolsey showed that this LP has an integrality gap of $\log(\max_{v\in \mathcal{N}} f(v))+1$ when $f$ is integer valued. 

%% file: eviction.tex
	
\section{Eviction Cost}
\label{sec:eviction}

In this section we show our algorithmic results for block-aware caching with respect to eviction costs. Our proof uses the (online) primal-dual method and hence requires an LP for the eviction cost model. It is straightforward to write a simple LP relaxation; unfortunately, the na\"{\i}ve relaxation has an integrality gap of $\Omega(\bsize)$ (see \cref{sec:basic_LP_formulation}), and recall that our goal is to beat the trivial algorithm's linear dependence on $\bsize$.

\subsection{Submodular Cover LP Formulation}

\label{sec:subcov_form}

To circumvent the na\"{\i}ve LP barrier, we strengthen the formulation using ideas from Wolsey's submodular cover LP. We start with some notation.

A \emph{flush} is a tuple $(B,t) \in \mathcal{B} \times [T]_0$. The flush $(B,t)$ corresponds to the event of evicting all cached pages of block $B$ at time $t$. (There is no reason to only evict some of them, since they can be fetched back for free.)  Let $S$ be a set of flushes. We say that a page $p$ is \textit{missing} at time $\tau$ according to $S$ if there exists $r(p,\tau) < t \leq \tau$ such that $(B(p),t) \in S$.\footnote{Adding to $S$ all flushes of the form $(B,0)$ ensures that also never-requested pages are missing by this definition.} Crucially, this definition ensures that the page $p_t$ requested at time $t$ is not missing at time $t$. We say than an algorithm is \emph{induced} by a set of flushes $S$ if the algorithm evicts all pages of block $B$ (except $p_t$) at time $t$ if and only if $(B,t) \in S$, and always loads $p_t$ at time $t$. Let $n_t$ be the number of pages requested up until time $t$.

We use the above to define a set function $f_\tau: 2^{\mathcal{B} \times [T]_0} \rightarrow \mathbb{Z}$ on sets of flushes:
\begin{align*}
	f_\tau(S) &:= \min(n-k, \ \left|\{p : p \text{ is \textit{missing} at time } \tau \text{ according to } S\}\right|)
\end{align*}

In words, $f_\tau(S)$ is the number of pages that are outside of the cache at time $\tau$ for the algorithm induced by $S$, where this number is capped at $n-k$. The algorithm induced by a set of flushes $S$ is feasible at time $\tau$ iff $f_\tau(S) \geq n-k$ for all $\tau$.

We show the following simple fact in \cref{sec:extra_proofs}:
\begin{restatable}{claim}{fsubmod}
	\label{claim:fsubmod}
	For every $\tau$, the function $f_\tau$ is submodular.
\end{restatable}
 
\begin{figure}
	\label{fig:ftau}
	\begin{mdframed}
		\centering
		\scalebox{0.65}{\input{ftau_figure.tikz}}
		\caption{Illustration of the function $f_\tau$. Each line represents a page, and each horizontal bar within the line represents an interval in which the page is not requested. Pages are grouped into their corresponding blocks. The solid vertical lines represent flushes. Suppose $n = 8$ and $k= 4$. Then $f_\tau(\{(B_1, t_1)\}) = 2$, $f_\tau(\{(B_2, t_2)\}) = 3$, but $f_\tau(\{(B_1, t_1), (B_2, t_2)\}) = 4$.}
	\end{mdframed}
\end{figure}

With the notation above, we can reformulate the block-aware caching problem with eviction cost as the solution to\footnote{This formulation is reminiscent of online and dynamic submodular cover problems \cite{gupta2020online,dynamicsubmod} in which the goal is also to maintain a feasible submodular cover while the underlying submodular function changes over time. However the cost models in these other works are very different.}
\begin{align}
	\label{eq:block_caching}
	\begin{array}{|rl|}
		\hline  & \\
		\displaystyle \min_{S \subseteq \mathcal{B} \times [T]_0} & \displaystyle \sum_{\substack{(B,t) \in S \\ t\geq 1}} c_B \\
		\text{subject to} & \\ & \\
		\forall \tau \in [T]: & f_\tau(S) \geq n-k. \\  & \\ \hline
	\end{array}
\end{align}
Note that because $f_\tau(S)$ counts the number of pages evicted by $S$ that are \emph{not} $p_\tau$, this single constraint captures both that the algorithm must flush at least $n-k$ pages in order to respect the cache size limit, and that the cache must contain page $p_t$ at time $t$.

Note as well that we allow the algorithm to perform flushes at time $0$, but only charge the cost for flushes performed after time $1$. This conveniently allows the algorithm to clear the cache initially at no extra cost.

Finally, we are ready to write our LP, which is is the intersection of the submodular cover LPs of \eqref{eq:wolsey_lp} for the functions $f_\tau$, across all time steps $\tau \in [T]$.
\begin{align}
	\begin{array}{|rl|}
		\multicolumn{2}{c}{\text{Primal}} \\ \hline  & \\
		\min & \displaystyle \sum_{B,t\geq 1} c_B \cdot \phi_B^t \\
		\text{subject to} & \\ & \\
		\begin{array}{c}
			\forall S\subseteq \mathcal{B} \times [T]_0, \\
			\forall \tau \in [T]
		\end{array}: & \begin{array}{l}
		\displaystyle \sum_{B,t} f_\tau((B,t) \mid S) \cdot \phi_B^t 
		\geq n-k-f_\tau(S)
	\end{array} \\ & \\
		\forall B \in \mathcal{B}, t \in [T]_0: & \phi_B^t \geq 0  \\  & \\ \hline
	\end{array} \label{eq:lp_p} \tag{P}
\end{align}
We will require the dual of this program, which is
\begin{align}
	\begin{array}{|rl|}
		\multicolumn{2}{c}{\text{Dual}} \\ \hline  & \\
		\max & \displaystyle \sum_{S, \tau} \left(n-k - f_\tau(S)\right) \cdot y^\tau_S \\
		\text{subject to} & \\ & \\
		\forall B \in \mathcal{B}, t \in [T]: & \displaystyle  \sum_{S, \tau} f_\tau((B,t) \mid S) \cdot y^\tau_S \leq c_B \\ & \\
		\begin{array}{c}
			\forall S\subseteq \mathcal{B} \times [T], \\
			\forall \tau \in [T]
		\end{array}: & y_S^\tau \geq 0 \\ \hline
	\end{array}  \label{eq:lp_d} \tag{D}
\end{align}

That \hyperref[eq:lp_p]{(P)} is a valid relaxation of \eqref{eq:block_caching} follows from \cref{lem:wolsey_integerpts}.

\begin{claim}[Corollary of \cref{lem:wolsey_integerpts}]
	\label{lem:integerpts}
	A set $S$ has $f_\tau(S) = n-k$ for all $\tau$ if and only if $\chi_S$, the
	characteristic vector of $S$, is a feasible integer solution to
	\hyperref[eq:lp_p]{(P)}.
\end{claim}

We note for intuition's sake that even the constraints \[\sum_{B,t} f_\tau((B,t)) \cdot \phi_B^t \geq n-k\] alone already avoid the bad integrality gap example of \cref{sec:basic_LP_formulation}. One reason is that truncating $f_\tau$ at $n-k$ prevents the LP from overestimating how much space will be saved by evictions.

Given an LP solution $\phi$, we also define the fractional value of a page $p$ missing from cache at time $t$ to be
\begin{align}x_p^t := \begin{cases}
	 1 & \text{if } r(p,t) = -\infty. \\
	 \min\left\{1,\sum_{u =r(p,t)+1}^{t} \phi_{B(p)}^u\right\} & \text{otherwise.} 
\end{cases} \label{eq:xpt_def}\end{align}
Intuitively, whenever some fraction $\delta$ of a flush $(B,t)$ is chosen, we imagine increasing the fractional amount by which each page in $B$ is evicted to extent $\delta$. On the other hand when a page $p_t$ is requested at time $t$, we reset its fractional value to $x^t_{p_t} = 0$.

\subsection{A $k$-Competitive Deterministic Online Algorithm}
\label{sec:deterministic}

Our first algorithmic result is a $k$-competitive deterministic online algorithm, which beats the trivial $k\bsize$ competitive ratio obtained by running the deterministic online algorithm for classical paging. 
Our deterministic algorithm is given in \cref{alg:onlinedet}. The algorithm constructs simultaneously a primal solution $x$ and a dual solution $y$ to the LPs~\eqref{eq:lp_p} and \eqref{eq:lp_d}. We will ensure that these solutions satisfy all constraints known up to that time. We use $C(\tau)$ for the set of pages in cache at time $\tau$, and we use $S$ for the set of flushes performed by the algorithm so far. At the start of the algorithm, $S$ is initialized as the set of all flushes of all blocks at time $0$; this amounts to clearing the initial cache.The primal solution $\phi$ is set to the characteristic vector of $S$, and dual solutions $y$ is initialized as the all-0-vector. At time $\tau$, we first add the requested page $p_\tau$ to the cache. If this violates the cache constraint, we continuously increase the dual variable $y_S^\tau$ corresponding to the current set $S$ and the current time $\tau$ until the dual constraint corresponding to some $(B,t)$ with $f_{\tau}((B,t)\mid S) \geq 1$ becomes tight. Once this happens, we evict all pages of block $B$ that are in cache (except $p_\tau$, in case it belongs to this block) and update $x$ and $S$ to reflect that the flush $(B,\tau)$ has been performed.

\begin{algorithm}
	\caption{Deterministic Online Algorithm}
	\label{alg:onlinedet}
	\begin{algorithmic}[1]	
		\State $S \leftarrow \{(B, 0) : B \in \mathcal{B} \}$\label{line:initS}
		\State $\phi \leftarrow \chi_S$, \ $y \leftarrow \vec 0$
		\For{time $\tau = 1, \ldots, T$}
		\State $C(\tau) \leftarrow C(\tau-1) \cup \{p_{\tau} \}$
		\If{$|C(\tau)| > k $}
		\State Increase $y^\tau_S$ until the dual constraint corresponding to some $(B,t)$ for which ${f_{\tau}((B,t)\mid S) \geq 1}$ is tight.
		\State $C(\tau) \leftarrow C(\tau)\setminus(B\setminus\{p_{\tau}\})$.
		\State $\phi_B^{\tau} \leftarrow 1$. \label{line:det_setto1}
		\State $S \leftarrow S \cup \{(B,\tau)\}$.
		\EndIf
		\EndFor
	\end{algorithmic}
\end{algorithm}

We show that both primal and dual solutions are feasible, and that the cost of the primal is at most $k$ times the cost of the dual. By weak duality, this implies:

\begin{theorem}
	\cref{alg:onlinedet} is $k$-competitive.
\end{theorem}

We begin by showing feasibility.
\begin{lemma} \label{det_primal_feasible}
\cref{alg:onlinedet} terminates. Upon termination, $\phi$ is feasible for~\eqref{eq:lp_p} and $y$ is feasible for~\eqref{eq:lp_d}.
\end{lemma}
\begin{proof}
	We first show that the algorithm terminates and the primal is feasible. Assume by induction that the algorithm maintains a feasible cache for every time step strictly less than $\tau$ (it is trivially feasible at time $0$). Since exactly one page is requested per time step, if the cache is not feasible at the beginning of time step $\tau$, then $|C(\tau)| = k+1$. If this is the case, then there must exist some $(B,t)$ such that $f_\tau((B,t) \mid S) \geq 1$, in which case the dual constraint corresponding to such a $(B,t)$ will become tight after $y_S^\tau$ is increased sufficiently (in particular, the increase of $y_S^\tau$ terminates). Note that $t\le \tau$ in this case. Then $f_\tau((B,\tau) \mid S) \geq f_\tau((B,t) \mid S) \geq 1$, so at least one page is evicted upon performing the flush $(B,\tau)$, and thus feasibility is restored at time $\tau$. Hence the algorithm maintains a feasible cache state for every time $\tau$, and by \cref{lem:integerpts} it follows that the primal is feasible for \eqref{eq:lp_p}.
	
	We now show feasibility of the dual. Clearly $y_S^\tau\ge 0$. Increasing $y_S^\tau$ could lead to a violation of the dual constraint corresponding to $(B,t)$ only if $f_\tau((B,t)\mid S)\ge 1$, but in this case we stop once the constraint becomes tight. Thus, dual constraints are never violated. Note that dual variables $y_{S}^\tau$ corresponding to future time steps $\tau$ are $0$, so it does not matter that the coefficients $f_\tau((B,t)\mid S)$ of future time steps are not known yet.
\end{proof}

Finally, we relate the primal and dual costs.

\begin{lemma}
	The cost of the primal is at most $k$ times the cost of the dual.
\end{lemma}
\begin{proof}
	The algorithm sets $\phi_B^{\tau} = 1$ if and only if $(B,\tau)\in S$, so the primal cost is
	\begin{align*}
		P &= \sum_{B, \tau} c_B \cdot  \phi_B^{\tau} = \sum_{(B,\tau)\in S} c_B.
	\end{align*}
	The algorithm adds $(B,\tau)$ to $S$ only if a constraint corresponding to some $(B,t)=(B,t_\tau)$ with $f_\tau((B,t_\tau)\mid S)\ge 1$ becomes tight at time $\tau$. Thus,
	\begin{align}
		c_B = \sum_{S', u} f_{u}((B,t_\tau) \mid S') \cdot y_{S'}^{u}.\label{eq:tightConstr}
	\end{align}
	To account for our primal cost, for every flush $(B,\tau)\in S$, we charge $f_u((B,t_\tau)\mid S')\cdot y^u_{S'}$ of the cost of the flush to the dual variable $y^u_{S'}$. Since every non-zero dual variable $y_{S'}^u$ in our final solution has its coefficient in the objective $n-k-f_u(S') \geq 1$ (otherwise it would never have been increased), it suffices to argue that each dual variable $y_{S'}^u$ receives a total charge of at most $k \cdot y_{S'}^u$. 
	
	To see this, note that $(B,\tau)$ only charges its cost to variables $y_{S'}^u$ for which $u \in [t_\tau,\tau]$. Indeed, $f_u((B,t_\tau)\mid S')>0$ only if $t_\tau\le u$; moreover, when \eqref{eq:tightConstr} becomes tight at time $\tau$ we have $y_{S'}^u=0$ for $u>\tau$ and all $S'$, and such $y_{S'}^u$ could subsequently increase only if $f_{u}((B,t_\tau) \mid S')=0$ since otherwise the dual constraint would become violated. After $(B,\tau)$ is added to $S$, for all $t' \le \tau$ and all $\tau' \geq \tau$ the multiplier $f_{\tau'}((B,t') \mid S) = 0$, and so $y_{S'}^u$ is charged at most once by every block $B$. Furthermore, if flush $(B,\tau)$ charges dual variable $y_{S'}^u$ then the coefficient $f_u((B,t)\mid S')$ is at most the number of pages from block $B$ that were in cache at the end of time step $u-1$. Since the total number of pages in cache at the end of time step $u-1$ was at most $k$, the total amount charged to $y_{S'}^u$ is at most $k \cdot y_{S'}^u$.
\end{proof}

\subsection{An $O(\log k)$-Competitive Monotone-Incremental Fractional Algorithm}

\label{sec:fractional}

In this section we give a competitive fractional algorithm for block caching with eviction cost. For simplicity of presentation, our algorithm is not online in the strict sense, as at time $\tau$ we allow it to change the value of $\phi_B^t$ for $t < \tau$. However, it has the crucial property that it only increases LP variables $\phi_B^t$. We call an algorithm with this property \emph{monotone-incremental}. This property suffices for our rounding procedure in \cref{sec:rounding} to yield an online algorithm. We prove:
\begin{theorem}
\label{thm:montone_incr}
There is an $O(\log k)$-competitive monotone-incremental fractional algorithm for block-aware caching with eviction cost.
\end{theorem}
To describe the algorithm, we define a flush $(B,t)$ to be \emph{alive} at time $\tau$ if $t=r(p,\tau)+1$ for some $p\in B$. Intuitively, for an offline algorithm it is most beneficial to flush a block $B$ only at time steps directly after some page from $B$ was requested. Accordingly, our fractional algorithm will increase $\phi_B^t$ only for $(B,t)$ that are alive. The algorithm is given in \cref{alg:onlinerand}. It starts by initializing $S$ as the set of all flushes at time $0$, $\phi$ as the corresponding characteristic vector, and $y$ as the all-$0$-vector. At time $\tau$, when some constraint $(S',\tau)$ is violated, we will show that this will also be the case for some $S'\supseteq S$, so that the condition of the while-loop will be true. We then increase the corresponding dual variable as well as primal variables corresponding to all alive flushes according to \eqref{eq:incrPhi}. While doing so, we occasionally add new flushes to the set $S$. As we will show later, all flushes $(B,t)$ added to the set $S$ will satisfy $\phi_B^t=1$, i.e., they are chosen integrally by the fractional algorithm.
\begin{algorithm}
	\caption{$O(\log k)$-competitive monotone-incremental fractional algorithm}
	\label{alg:onlinerand}
	\begin{algorithmic}[1]	
		\State $S \leftarrow \{(B, 0) : B \in \mathcal{B} \}$.
		\State $\phi \leftarrow \chi_S$, \ $y \leftarrow \vec 0$.
		\For{time $\tau = 1, \ldots, T$}
		\While{$\exists S' \supseteq S$ s.t. primal constraint of $(S',\tau)$ violated} \label{alg:frac_startwhile}
		\State Increase $y^\tau_{S'}$ continuously, and meanwhile for every alive $(B,t)$ increase $\phi_B^{t}$ at rate \begin{align}\frac{d\phi_B^t}{dy_{S'}^\tau} = 
				 \frac{\ln(k \cdot \bsize+1)}{c_B} \cdot   f_\tau ((B,t) \mid S') 
 \cdot  \left(\phi_B^{t} + \frac{1}{k\cdot \bsize}\right)
			\label{eq:incrPhi}
		\end{align} until the dual constraint corresponding to some alive $(B_0,t_0)$ for which the inequality $f_\tau((B_0,t_0)\mid S')\ge 1$ is tight. \label{line:frac_roc}
		\State $S \leftarrow S \cup \{(B_0,t_0)\}$. \label{alg:frac_endwhile}
		\EndWhile
		\EndFor
	\end{algorithmic}
\end{algorithm}
\begin{lemma}
	\cref{alg:onlinerand} terminates.
\end{lemma}
\begin{proof}
	If the condition of the while-loop is true, then $f_\tau(S')<n-k$, and thus there exists some alive $(B_0,t_0)$ with $f_\tau((B_0,t_0)\mid S')\ge 1$. Increasing $y_{S'}^\tau$ sufficiently will eventually tighten a corresponding constraint, so each iteration of the while-loop terminates. When a new element is added to $S$ at the end of an iteration, $f_\tau(S)$ increases by at least $1$. When $f_\tau(S)$ has reached value $n-k$ (or earlier), the condition of the while-loop cannot be true any more.
\end{proof}

\begin{lemma}
	\label{lem:frac_pd_feasible}
	At the end of \cref{alg:onlinerand}, $\phi$ is feasible for~\eqref{eq:lp_p} and $y$ is feasible for~\eqref{eq:lp_d}.
\end{lemma}

Before proving \cref{lem:frac_pd_feasible}, we need the following claim, whose proof we leave for \cref{sec:extra_proofs}.
\begin{definition}
	Given a fractional solution $\phi$, we say that the constraint $(S, \tau)$ is maximal-integral if for any $(B,t)$ such that $\phi_B^t = 1$ it holds that $(B,t) \in S$.
\end{definition}

\begin{restatable}{claim}{smaxint}
	\label{claim:S_max}
	If a fractional solution $\phi$ to \eqref{eq:lp_p} has no violated maximal-integral constraints, then $\phi$ is feasible.
\end{restatable}

\begin{proof}[Proof of \cref{lem:frac_pd_feasible}]
	We first show feasibility of the dual. Suppose the dual constraint corresponding to some $(B,t)$ gets violated when $y_{S'}^\tau$ is increased. Let $t_0\le t$ be maximal such that $(B,t_0)$ is alive at time $\tau$. (If no such $t_0$ exists, then any pages evicted by the flush $(B,t)$ are requested again in $[t,\tau]$; but then $f_\tau((B,t)\mid S')=0$, so increasing $y_{S'}^\tau$ would not have led to a violation of the constraint corresponding to $(B,t)$.) Since no pages of $B$ are requested at times in $[t_0,t)$, we have $f_\tau((B,t_0)\mid S'')=f_\tau((B,t)\mid S'')$ for any $S''$, meaning that the dual constraint of $(B,t_0)$ would become violated at the same time during the increase of $y_{S'}^\tau$. But then $f_\tau((B,t_0)\mid S')\ge 1$ (otherwise, increasing $y_{S'}^\tau$ would not increase the left-hand side of constraint $(B,t_0)$) and therefore we would have stopped increasing $y_{S'}^\tau$ when the constraint got tight.
	
	To see that the primal is feasible, we will show that for all $(B,t)\in S$ we have $\phi_B^t=1$. It then follows from \cref{claim:S_max} in the appendix that if a primal constraint $(S',\tau)$ is infeasible, then this is also the case for some $S'\supseteq S$; thus, the algorithm would not have terminated.
	
	Consider the differential equation $dz / dy = \eta\cdot(z + \delta)$ for some constants $\eta\ge 0$ and $\delta>0$. When $y$ increases from $a$ to $b$, we have
	\begin{align}
	\ln(z(b) + \delta) - \ln(z(a) + \delta) = \eta\cdot(b-a).\label{eq:frac_ode}
	\end{align}
	For some $(B,t)$ that eventually gets added to $S$, consider the dynamics of $\phi_B^t$. It starts at $0$, and increases with every $y_{S'}^\tau$ according to \eqref{eq:incrPhi}. Applying \eqref{eq:frac_ode} for every such $y_{S'}^\tau$ and summing, we have
	\begin{align*}
		&\ln\left(\phi_B^t + \frac{1}{k\cdot \bsize}\right) - \ln \left(\frac{1}{k\cdot \bsize}\right) = \sum_{S',\tau} \frac{\ln(k \cdot \bsize+1)}{c_B}  \cdot  f_\tau((B,t) \mid S') \cdot  y_{S'}^\tau. \\
		\intertext{Taking exponents and solving, we have that} 
		\phi_B^{t} &= \frac{1}{k\cdot \bsize}\cdot \left( \exp\left(\frac{\ln(k \cdot \bsize+1)}{c_B} \sum_{S',\tau} f_\tau((B,t) \mid S') \cdot y_{S'}^\tau\right) -1 \right).
	\end{align*}
	In particular, when constraint $(B,t)$ becomes tight and is added to $S$, the value of $\phi_B^t$ is $1$.
\end{proof}

\begin{lemma}
	\label{lem:frac_cost}
	The cost of the primal is at most $O(\log k)$ times the cost of the dual.
\end{lemma}
\begin{proof}
	Consider the algorithm at a fixed time $\tau$ during a step in which $y_{S'}^\tau$ is increased by an infinitesimal amount $dy_{S'}^\tau$. The dual profit is $dy_{S'}^\tau \cdot (n-k-f_\tau(S'))$, and so it suffices bound the corresponding increase in primal cost. The primal cost increase is:
	\begin{align*}
		&\sum_{B,t} c_B \cdot \frac{1}{c_B} \ln(k \cdot \bsize+1) \cdot f_\tau((B,t) \mid S') \cdot \left(\phi_B^t + \frac{1}{k\cdot \bsize}\right) \cdot dy_{S'}^\tau \\
		&= \sum_{B,t} \ln(k \cdot \bsize+1)\cdot f_\tau((B,t) \mid S') \cdot \left(\phi_B^t + \frac{1}{k\cdot \bsize}\right) \cdot dy_{S'}^\tau.
	\end{align*}
	Since we only increase $y_{S'}^\tau$ if the corresponding constraint in the primal is not satisfied, we have
	\begin{align}
		\sum_{B,t} f_\tau((B,t) \mid S') \cdot \phi_B^t &< n - k - f_\tau(S'). \label{eq:frac_primalunsat} \\
		\intertext{We claim that}
		\sum_{(B,t)\text{ alive}}  \frac{f_\tau((B,t) \mid S')}{k\cdot \bsize} &\leq n - k - f_\tau(S'). \label{eq:frac_1overkl}
	\end{align}
	
	Inequalities \eqref{eq:frac_primalunsat} and \eqref{eq:frac_1overkl} together imply that the increase in the primal cost is at most $2\ln(k \cdot \bsize+1)\cdot(n-k-f_\tau(S))\cdot  dy_S^\tau$, which in turn implies the lemma statement since $\bsize\le k$.
	
	To prove \eqref{eq:frac_1overkl}, we first show that 
	\begin{align}\sum_B \frac{f_\tau ((B,\tau) \mid S')}{k} \leq n  - k - f_{\tau}(S').\label{eq:frac_divminineq}\end{align}
	Note that $\sum_B f_\tau ((B,\tau) \mid S') \leq n - f_{\tau}(S') - 1$. In the case that $n - f_{\tau}(S') \geq k+1$, \eqref{eq:frac_divminineq} holds by the fact that $(z-1)/k \leq z - k$ for every $z\geq k+1$. Otherwise, when $n - f_{\tau}(S') < k+1$, then $f_\tau(S') = n-k$ (since $f_\tau$ is integer valued and truncated at $n-k$). Then \eqref{eq:frac_divminineq} holds with both sides equal to $0$.
	
	We now obtain \eqref{eq:frac_1overkl} via
	\begin{align*}
		\sum_{(B,t)\text{ alive}}  \frac{f_\tau((B,t) \mid S')}{k\cdot \bsize} &\leq \sum_{(B,t)\text{ alive}}  \frac{f_\tau((B,\tau) \mid S')}{k\cdot \bsize} \\
		&\leq \sum_{B}  \frac{f_\tau((B,\tau) \mid S')}{k}\\
		&\leq n - k - f_\tau(S')
	\end{align*}
	where the second inequality uses that there are at most $\bsize$ flushes alive at any time, and the last inequality is \eqref{eq:frac_divminineq}.
\end{proof}

\subsection{An $O(\log \Delta k)$-Competitive Online Randomized Rounding Scheme}

\label{sec:rounding}

Finally we show an online $O(\log \Delta k)$ randomized rounding scheme for our block caching LP \eqref{eq:lp_p}. Recall the definition of the aspect ratio $\Delta$ from \cref{prelim} and note that in the standard unweighted setting, $\Delta = 1$.

At time $t$, the algorithm evicts block $B$ with probability $\gamma \cdot \phi_B^\tau$, where $\gamma = O(\log k \Delta)$. If the cache is still infeasible at time $t$, evict an arbitrary block so long as at least one of its pages has $x_p^t > 0$ (recall from the definition \eqref{eq:xpt_def} that $x_p^t$ is the amount missing from page $p$ at time $t$).  For clarity of exposition, the rounding procedure is written as if the underlying fractional solution $(x,\phi)$ is computed online. However the procedure can be carried out so long as the solution is monotone-incremental; at time $\tau$, if the fractional solution increases any $\phi_B^t$ for $t < \tau$ by some amount $\delta_t$, we can evict $B$ at time $\tau$ with probability $\min(1, \gamma\cdot(\phi_B^\tau + \sum_{t < \tau} \delta_t))$. 

Thus together with \cref{thm:montone_incr}, our rounding scheme implies:

\begin{theorem}
For block-aware caching with eviction cost, there exists an $O(\log k \log (k\Delta))$-competitive algorithm.
\end{theorem}

Furthermore, using the round-or-separate procedure of \cite{gupta2020online}, one can simultaneously solve and round the \subcov LP \eqref{eq:wolsey_lp} offline in polynomial time. Using the analysis of this section, this implies:
\begin{theorem}
For block-aware caching with eviction cost, there exists an $O(\log (k\Delta))$-approximation algorithm.
\end{theorem}

We perform the rounding assuming a few key properties of our fractional solution which we show we can assume (online) without changing our asymptotic guarantees.

\begin{restatable}{lemma}{structsol} \label{lem:structured_solution}
    Let $(x,\phi)$ be a fractional solution for LP \eqref{eq:lp_p}. For an additional multiplicative constant factor to the competitive ratio, we can assume that $(x,\phi)$ has the following properties:
    \begin{itemize}
        \item For every time $t$, every page $p$ has $x_p^t \in [0,\nicefrac{1}{2}] \cup \{1\}$.
        \item Every nonzero coordinate has $\phi_B^t \geq \nicefrac{1}{4k^2}$.
    \end{itemize}
\end{restatable}

We defer the proof to \cref{sec:extra_proofs}.

A key technical tool in this section is a lemma from \cite{gupta2020online}, which in turn relies on a relationship between continuous extensions of submodular functions proven by \cite{vondrak2007submodularity}.

\begin{lemma}[Lemma 2.5 of \cite{gupta2020online}]
	\label{lem:gl_rounding}
	Let $x \in [0,1]^n$ be a feasible solution to \eqref{eq:wolsey_lp}. Let $R$ be a set obtained by performing randomized rounding according to $\min(1,\gamma \cdot x)$. Then: 
	\[\expectover{R}{f(R)} \geq f(\mathcal{N}) - e^{-\gamma} f(\mathcal{N}).\]
\end{lemma}

We can now present our rounding scheme. 

\begin{algorithm}
	\caption{$O(\log k)$-Approximate Rounding}
	\label{alg:onlineround}
	\begin{algorithmic}[1]
		\For {time $\tau \in [T]$}
		\State For every block $B$ evict the set $\{p\in B \mid x_p^t > 0\}$ with probability $\min(1, \gamma \cdot \phi_B^t)$. \label{line:r-round-step}
		\State Fetch $p_\tau$ if it is missing from the cache.
		\While{the cache is infeasible} \label{line:evict_loop}
		\State Let $B$ be an arbitrary block in the cache that has a page $p$ with $x_p^t > 0$, evict the set of pages $\{p\in B \mid x_p^t > 0\}$. \label{line:evict_loop2}
		\EndWhile 
		\EndFor
	\end{algorithmic}
\end{algorithm}

Note that we assume that the fractional solution on which \cref{alg:onlineround} executes is one that has the properties given by \cref{lem:structured_solution}.

We now prove our main rounding lemma.
\begin{lemma}
	\label{lem:rounding2}
	For $\gamma = \log (4k^2 \bsize\Delta)$, given a feasible fractional solution $(x,\phi)$ with cost $c(\phi)$, \cref{alg:onlineround} produces a feasible integral cache policy of cost $O(\log k\Delta) \cdot c(\phi)$.
\end{lemma}

To prove \cref{lem:rounding2}, we charge the cost of the algorithm to the fetching cost of the fractional solution. To relate this fractional fetching cost to the fractional eviction cost, we need a claim which we prove in \cref{sec:extra_proofs}. 
\begin{restatable}{claim}{loading}
	\label{claim:loading}
	Let $c_{\textsc{Fetch}}(z)$ be the fetching cost of a fractional solution $z$. Then
	\[c_{\textsc{Fetch}}(z) \leq \bsize \left( c(z) + \sum_{B \in \mathcal{B}} c_B\right).\]
\end{restatable}

\begin{proof}[Proof of \cref{lem:rounding2}]
	Let $x,\phi$ be a fractional solution given by \cref{lem:structured_solution}, and let $S$ be the set of flushes performed by our algorithm.
	
	The algorithm produces a feasible cache policy by construction, as we always fetch $p_t$ and we always run the eviction loop in \cref{line:evict_loop,line:evict_loop2} until the cache is feasible. Note that there is always a block to evict with a page $p$ that has $x_p > 0$, otherwise $x$ is integral, and is the characteristic vector of the pages the algorithm has in cache, in which case the algorithm's cache is already feasible since the fractional solution is feasible. Furthermore, the expected cost of the evictions due to the randomized rounding step at \cref{line:r-round-step} is at most $\gamma \cdot c(\phi) = O(\log k\Delta)\cdot c(\phi)$. 
	
	It remains to show that the total cost due to alterations in the eviction loop in \cref{line:evict_loop,line:evict_loop2} is bounded. We now show that it is at most $O(c(\phi))$.
	
	Let $\Lambda$ be the set of times $\tau$ such that  $p_\tau$ is not already fully in the fractional cache. Our algorithm maintains the invariant that if $x_p^t = 0$, then it is also fully in cache of the integral solution produced by our algorithm at time $t$. This means that at times $\tau \not\in \Lambda$, neither the fractional solution nor the rounding algorithm incur a cost increase. Hence we focus on the case where $\tau \in \Lambda$.
	
	For every $\tau$, the solution $\phi$ is feasible for the LP \eqref{eq:wolsey_lp} with the function $f^\tau$, so by \cref{lem:gl_rounding}
	\[\expect{f_\tau(S)} \geq n - k -\frac{1}{4k^2 \bsize\Delta}.\]
	In particular, this holds for all $\tau \in \Lambda$. In words, the expected number of pages in cache is bounded by $k + \frac{1}{4k^2 \bsize\Delta}$. Since every eviction due to \cref{line:evict_loop2} costs at most $\cmax$ and evicts at least one page, the expected cost of the alteration while loop at time $\tau$ is bounded by $\cmax / (4k^2 \bsize\Delta) = \cmin / 4k^2 \bsize$.
	
	On the other hand, since $\tau \in \Lambda$, the page $p_\tau$ is not fully in cache, and since by \cref{lem:structured_solution} the fractional solution evicts pages in increments of at least $1/(4k^2)$, it holds that $x_{p_\tau}^\tau \geq 1/(4k^2)$. This means that the \emph{fetching} cost of the fractional solution at time $\tau$ is at least $\cmin / 4k^2$. 
	
	Hence the expected cost of the alteration step in time $\tau$ is at most the fractional fetching cost at time $\tau$, divided by $\bsize$. Summing this inequality over time, the total cost paid by the algorithm over all all time due to \cref{line:evict_loop2} is at most $c_{\textsc{Fetch}}(\phi) / \bsize$. By \cref{claim:loading}, the fetching cost $c_{\textsc{Fetch}}(\phi) \leq \beta(c(\phi) + \sum_{B \in \mathcal{B}} c_B)$, and hence the total cost of alterations is at most $c(\phi) + \sum_{B \in \mathcal{B}} c_B$. This completes the proof.
\end{proof}

%% file: ftau_figure.tikz
\tikzset{every picture/.style={line width=0.75pt}} 

\begin{tikzpicture}[x=0.75pt,y=0.75pt,yscale=-1,xscale=1]
	
	\draw  [draw opacity=0][fill={rgb, 255:red, 255; green, 0; blue, 0 }  ,fill opacity=0.6 ][line width=0.75]  (90,80) -- (210,80) -- (210,90) -- (90,90) -- cycle ;
	\draw  [draw opacity=0][fill={rgb, 255:red, 255; green, 0; blue, 0 }  ,fill opacity=0.6 ][line width=0.75]  (90,100) -- (310,100) -- (310,110) -- (90,110) -- cycle ;
	\draw  [draw opacity=0][fill={rgb, 255:red, 255; green, 0; blue, 0 }  ,fill opacity=0.6 ][line width=0.75]  (90,120) -- (120,120) -- (120,130) -- (90,130) -- cycle ;
	\draw  [draw opacity=0][fill={rgb, 255:red, 255; green, 0; blue, 0 }  ,fill opacity=0.6 ][line width=0.75]  (90,140) -- (170,140) -- (170,150) -- (90,150) -- cycle ;
	\draw  [draw opacity=0][fill={rgb, 255:red, 0; green, 37; blue, 255 }  ,fill opacity=0.6 ] (90,170) -- (160,170) -- (160,180) -- (90,180) -- cycle ;
	\draw  [draw opacity=0][fill={rgb, 255:red, 0; green, 37; blue, 255 }  ,fill opacity=0.6 ] (90,190) -- (230,190) -- (230,200) -- (90,200) -- cycle ;
	\draw  [draw opacity=0][fill={rgb, 255:red, 0; green, 37; blue, 255 }  ,fill opacity=0.6 ] (90,210) -- (100,210) -- (100,220) -- (90,220) -- cycle ;
	\draw  [draw opacity=0][fill={rgb, 255:red, 0; green, 37; blue, 255 }  ,fill opacity=0.6 ] (90,230) -- (180,230) -- (180,240) -- (90,240) -- cycle ;
	\draw  [draw opacity=0][fill={rgb, 255:red, 255; green, 0; blue, 0 }  ,fill opacity=0.6 ][line width=0.75]  (220,80) -- (430,80) -- (430,90) -- (220,90) -- cycle ;
	\draw  [draw opacity=0][fill={rgb, 255:red, 255; green, 0; blue, 0 }  ,fill opacity=0.6 ][line width=0.75]  (130,120) -- (470,120) -- (470,130) -- (130,130) -- cycle ;
	\draw  [draw opacity=0][fill={rgb, 255:red, 255; green, 0; blue, 0 }  ,fill opacity=0.6 ][line width=0.75]  (320,100) -- (400,100) -- (400,110) -- (320,110) -- cycle ;
	\draw  [draw opacity=0][fill={rgb, 255:red, 0; green, 37; blue, 255 }  ,fill opacity=0.6 ] (170,170) -- (470,170) -- (470,180) -- (170,180) -- cycle ;
	\draw  [draw opacity=0][fill={rgb, 255:red, 0; green, 37; blue, 255 }  ,fill opacity=0.6 ] (240,190) -- (450,190) -- (450,200) -- (240,200) -- cycle ;
	\draw  [draw opacity=0][fill={rgb, 255:red, 0; green, 37; blue, 255 }  ,fill opacity=0.6 ] (110,210) -- (470,210) -- (470,220) -- (110,220) -- cycle ;
	\draw  [draw opacity=0][fill={rgb, 255:red, 0; green, 37; blue, 255 }  ,fill opacity=0.6 ] (360,230) -- (470,230) -- (470,240) -- (360,240) -- cycle ;
	\draw [color={rgb, 255:red, 255; green, 0; blue, 0 }  ,draw opacity=1 ][line width=3]    (208,50) -- (210,270) ;
	\draw [color={rgb, 255:red, 0; green, 37; blue, 255 }  ,draw opacity=1 ][line width=3]    (290,50) -- (290,270) ;
	\draw [color={rgb, 255:red, 0; green, 0; blue, 0 }  ,draw opacity=1 ][line width=2.25]  [dash pattern={on 2.53pt off 3.02pt}]  (410,50) -- (410,270) ;
	\draw  [draw opacity=0][fill={rgb, 255:red, 255; green, 0; blue, 0 }  ,fill opacity=0.6 ][line width=0.75]  (180,140) -- (470,140) -- (470,150) -- (180,150) -- cycle ;
	\draw  [draw opacity=0][fill={rgb, 255:red, 255; green, 0; blue, 0 }  ,fill opacity=0.6 ][line width=0.75]  (440,80) -- (470,80) -- (470,90) -- (440,90) -- cycle ;
	\draw  [draw opacity=0][fill={rgb, 255:red, 255; green, 0; blue, 0 }  ,fill opacity=0.6 ][line width=0.75]  (410,100) -- (470,100) -- (470,110) -- (410,110) -- cycle ;
	\draw  [draw opacity=0][fill={rgb, 255:red, 0; green, 37; blue, 255 }  ,fill opacity=0.6 ] (460,190) -- (470,190) -- (470,200) -- (460,200) -- cycle ;
	\draw  [draw opacity=0][fill={rgb, 255:red, 0; green, 37; blue, 255 }  ,fill opacity=0.6 ] (190,230) -- (350,230) -- (350,240) -- (190,240) -- cycle ;
	
	\draw (49,106) node [anchor=north west][inner sep=0.75pt]  [font=\LARGE]  {$\textcolor[rgb]{1,0,0}{B}\textcolor[rgb]{1,0,0}{_{1}}$};
	\draw (49,195) node [anchor=north west][inner sep=0.75pt]  [font=\LARGE]  {$\textcolor[rgb]{0,0.15,1}{B}\textcolor[rgb]{0,0.15,1}{_{2}}$};
	\draw (183,15) node [anchor=north west][inner sep=0.75pt]  [font=\LARGE]  {$\textcolor[rgb]{1,0,0}{(}\textcolor[rgb]{1,0,0}{B}\textcolor[rgb]{1,0,0}{_{1}}\textcolor[rgb]{1,0,0}{,t}\textcolor[rgb]{1,0,0}{_{1}}\textcolor[rgb]{1,0,0}{)}$};
	\draw (264,15) node [anchor=north west][inner sep=0.75pt]  [font=\LARGE]  {$\textcolor[rgb]{0,0.15,1}{(}\textcolor[rgb]{0,0.15,1}{B}\textcolor[rgb]{0,0.15,1}{_{2}}\textcolor[rgb]{0,0.15,1}{,t}\textcolor[rgb]{0,0.15,1}{_{2}}\textcolor[rgb]{0,0.15,1}{)}$};
	\draw (405,15) node [anchor=north west][inner sep=0.75pt]  [font=\LARGE]  {$\tau$};

\end{tikzpicture}

%% file: loading.tex
	\section{Fetching Cost}

\label{sec:loading}

We present our $\Omega(\beta)$ lower bound against randomized algorithms for online block-aware caching with fetching costs. We first present a bicriteria rounding algorithm for the na\"{\i}ve LP of \cref{sec:basic_LP_formulation}. We then argue that this procedure can be used to derandomize any randomized algorithm for block-aware caching with fetching costs. Together with the lower bound against deterministic algorithms given by \cite{beckmann2021brief}, this implies a lower bound against randomized algorithms. 

\subsection{Bicriteria Online Rounding Algorithm}

\label{sec:hk_rounding}

	Consider the following deterministic online rounding scheme. For every page $p$, evict $p$ from the cache at time $t$ if $x_p^t > \nicefrac{1}{2}$. If a page $p_t$ is not in cache upon request at time $t$, then at time $t$ fetch all pages from $B(p_t)$ such that $x_p^t \leq 1/2$.

\begin{theorem}
	\label{thm:hk_rounding}
	Given a feasible fractional solution $x$ to the block-aware caching problem, the procedure above produces an integral solution that uses at most $2k$ cache space at any point in time, and whose fetching cost is at most twice the fetching cost of $x$.
\end{theorem}

\begin{proof}
	The procedure produces a feasible solution by construction, since $x_{p_t}^t = 0 \leq \nicefrac{1}{2}$. It also violates the cache size constraint by at most a factor of $2$, since no page is present in the integral cache unless $x_p^t \leq \nicefrac{1}{2}$, meaning the fractional cache usage is at least half the integral cache usage.
	
	Finally, to justify that the integral solution has cost at most twice the fractional cost, charge the cost of integrally loading $B(p_t)$ to the fractional decrease of $x_B^t$ since the last time $t'$ at which $B(p_t)$ was loaded. Since $p_t$ had $x_p^{t'} > \nicefrac{1}{2}$ (otherwise we would have loaded it earlier), the fractional cost incurred since time $t'$ was at least $\nicefrac{1}{2} \cdot c_{B(p)}$.
\end{proof}

\begin{corollary}
	When $k=2h$, there is a 2-competitive offline algorithm for block-aware caching with fetching cost.
\end{corollary}

We mention briefly that a similar rounding procedure produces a cache policy that is $2$-competitive with the eviction cost of the fractional solution, and also uses at most a factor $2$ more space. If $p_t$ is not in cache at time $t$, fetch it. On the other hand if any page in cache at time $t$ has fractional value $x_p^t < \nicefrac{1}{2}$, evict all of $B(p)$.

\subsection{Lower Bounds for Randomized Algorithms}

Finally we turn to showing our lower bound. Our starting point is the lower bound of \cite{beckmann2021brief} against deterministic algorithms.

\begin{theorem}[Theorem 4.1 of \cite{beckmann2021brief}]
	\label{thm:charlie_lb}
	The competitive ratio of any deterministic online policy for block-aware caching with fetching costs is at least 
	\[\frac{k+(B-1)(h-1)}{k-h+1}\]
	for $h \leq k-B + 1$. 
\end{theorem}

We now show how to use the online deterministic rounding procedure of \cref{sec:hk_rounding} to derandomize any online algorithm for Block-Aware caching with fetching costs. This proves the main claim of this section:

\begin{theorem}
	The competitive ratio of any randomized policy for block-aware caching with fetching costs is at least 
	\[\frac{2k+(B-1)(h-1)}{4k-2h+2}\]
	for $h \leq k-B + 1$. 
\end{theorem}

\begin{proof}
	Suppose there is a randomized online algorithm $\mathcal{R}$ for $(h,k)$-block-aware caching with fetching costs with expected cost $c_\mathcal{R}$. Then we can convert this randomized cache policy online to a fractional solution $x$. To do so, set $x^t_p$ be the expected value of the indicator of whether page $p$ is loaded at time $t$. Note that these expectations can be computed using only the sequence of requests up to and including time $t$. This solution $x$ is feasible to the simple fetching cost LP \eqref{eq:naive_lp}, and furthermore has LP cost $c_\mathcal{R}$. 
	
	Applying \cref{thm:hk_rounding} to the fractional solution $x$ produces an integral cache policy cost at most $2 \cdot c_\mathcal{R}$ and space $2k$.  The claim follows by using the lower bound on the cost of any such policy given by \cref{thm:charlie_lb}, and solving for $c_{\mathcal{R}}$.
\end{proof}

Combining this with the well known $\Omega(\log k)$ lower bound for randomized algorithms for classical paging, we obtain the following consequence.
\begin{corollary}
	When $k= O(h)$, no randomized algorithm has competitive ratio better than $\Omega(B + \log k)$.
\end{corollary}

%% file: appendix.tex
\appendix

\section{Appendix}

\label{sec:appendix}

\subsection{Deferred Proofs}

\label{sec:extra_proofs}

\betaoff*
\begin{proof}
	
	Consider the following instance. For any $\bsize$, let $n = 2\bsize^2$ pages be organized into $2\bsize$ blocks of size $\bsize$. Let $P$ be the first $\bsize$ blocks and $Q$ be the second $\bsize$ blocks. We set $k=\bsize^2$, and fill the cache initially with all the pages of the $P$ blocks. The request sequence consists of rounds. For $i=1,\dots,\bsize$, in round $i$ request the first $\bsize-i$ pages of each $P$ block, and the first $i$ of the $Q$ blocks in their entirety, and repeat this sequence $L$ times within the round. For sufficiently large constant $L$, the optimal solution must have precisely the requested pages of a round in its cache. Thus, in round $i$ it evicts the $(\bsize-i+1)^{th}$ page of each $P$ block and fetches the $i^{th}$ $Q$ block in its entirety. The fetching cost of this solution is $\bsize$, while the eviction cost is $\bsize^2$.
	
	To see the other direction, observe that if we instead start the cache with the pages of the $Q$ blocks, and for $i=1,\dots,\bsize$ we make round $i$ request precisely the pages \emph{not} requested in round $i$ above (and once again repeat this sequence $L$ times), the optimal solution will always evict one $Q$ block in its entirety and fetch a single page from each $P$ block in each round. Here the fetching cost is $\bsize^2$ and the eviction cost is $\bsize$.
\end{proof}

\fsubmod*

\begin{proof}
	Consider the function $g^\tau$, where 
	\begin{align*}
		g^\tau(S) &:= \left|\{p : p \text{ is \textit{missing} at time } \tau \text{ according to } S\}\right| \\
		&= \left| \bigcup_{\phi \in S} \{p : p \text{ is \textit{missing} at time } \tau \text{ according to } \phi\}\right|.\end{align*}
	$g^\tau$ is a coverage function, and hence it is submodular.
	The function $f^\tau$ is the minimum of $g^\tau$ and the constant function $n-k$, and so $f^\tau$ is also submodular.
\end{proof}

\smaxint*
\begin{proof}
	It suffices to show that if $\phi$ violates a constraint $(S,\tau)$, and $(B_0,t_0)$ is such that $\phi_{B_0}^{t_0} = 1$, then $x$ also violates $(S \cup \{(B_0,t_0)\}, \tau)$.  Since $x$ is violated,
	\begin{align*}
		n - k - f_\tau(S) &> \sum_{B,t} f_\tau((B,t) \mid S) \cdot \phi_{B}^t \\
		&= \sum_{(B,t) \neq (B_0,t_0)} f_\tau((B,t) \mid S) \cdot \phi_{B}^t + f_\tau((B_0, t_0) \mid S)\\
		&\geq \sum_{(B,t) \neq (B_0,t_0)} f_\tau((B,t) \mid S \cup \{(B_0, t_0)\}) \cdot \phi_{B}^t \\
		& \quad +  f_\tau((B_0, t_0) \mid S)
	\end{align*}
	where the second inequality above used submodularity. Rearranging gives that
	\[\sum_{(B,t) \neq (B_0,t_0)} f_\tau((B,t) \mid S \cup \{(B_0, t_0)\}) \cdot \phi_{B}^t < n - k - f_\tau(S \cup \{(B_0, t_0)\}). \qedhere\]
\end{proof}

\structsol*

\begin{proof}
	To guarantee the first property, every time a page from a block $B$ is evicted to extent $\nicefrac{1}{2}$, evict the entire block $B$ for a cost of $c_B$. Charge this eviction to the evictions that caused this page to go from $0$ to $1/2$, which cost at least $c_B/2$.
	
	To ensure the second property, consider the following online algorithm. 
	
	\begin{algorithm}
		\caption{Structure Solution}
		\label{alg:structsol}
		\begin{algorithmic}[1]	
			\State Define solution $\widetilde x$ such that $\widetilde \phi_B^t = \phi_B^t + \mathbbm{1} \{\phi_B^t \geq  \nicefrac{1}{2}\} \cdot (1 - \phi_B^t)$.
			\For{block $B$}
			\State Set $t_B \leftarrow 0$.
			\For{time $t \in [T]$}
			\State Let $\Delta =  \sum_{t' = t_B + 1}^t \widetilde \phi_{B}^{t'}$.
			\If{$\Delta  \geq \nicefrac{1}{4k^2}$}
			\State $\varphi_B^t \leftarrow \Delta$.
			\State $t_B \leftarrow t$. \label{line:tb_set}
			\EndIf
			\EndFor
			\EndFor
			\State Output $\widetilde \varphi = \min(2 \cdot \varphi, 1)$.
		\end{algorithmic}
	\end{algorithm}
	The structural guarantee that every nonzero coordinate has $\varphi_B^t \geq \nicefrac{1}{4k^2}$ holds by construction. The cost is also less than $2 \cdot c(\phi)$ by construction.
	
	It remains to show $\varphi$ is feasible. Consider any constraint of the form:
	\[\sum_{B,t} f_\tau((B,t) \mid S) \cdot \phi_B^t \geq n - k - f_\tau(S)\]
	For a block $B$, let $\tau_B$ be the last time before $\tau$ that $t_B$ was set to in \cref{line:tb_set}. By construction 
	\begin{align*}
		\sum_{t = \tau_B+1}^\tau \phi_B^t  \leq \frac{1}{4k^2}. 
	\end{align*}
	Since $f((B,t) | S) \leq k$, and by the property that every $p$ has $x_p^t \in [0,\nicefrac{1}{2}]\cup\{1\}$, there are at most $2k$ blocks with nonzero pages in cache. This also means
	\begin{align*}
		\sum_B \sum_{t = \tau_B+1}^\tau f((B,t) \mid S) \cdot \phi_B^t  \leq \frac{1}{2}
	\end{align*}
	and hence
	\begin{align*}
		\sum_{B,t} f_\tau ((B,t) \mid S) \cdot \varphi_B^t \geq n - k - f_\tau(S) - \frac{1}{2}.
	\end{align*}

	If $n - k - f_\tau(S) = 0$, then the constraint is also trivially satisfied by $\varphi$. Else $n - k - f_\tau(S) \geq 1$. To conclude, note that since $S$ is maximal-integral, and $\phi$ has no coordinates $\phi_B^t \in (\nicefrac{1}{2}, 1)$:
	\begin{align*}
		\sum_{t\leq \tau_0} \sum_B f((B,t) \mid S) \cdot \varphi_B^t &= 2\sum_{t\leq \tau_0} \sum_B f((B,t) \mid S) \cdot \phi_B^t \\
		&\geq 2\left(n - k - f_\tau(S) - \frac{1}{2} \right) \\
		&\geq n - k - f_\tau(S)
	\end{align*}
	If $\varphi$ satisfies all integral-maximal constraints, \cref{claim:S_max} implies it also satisfies all other constraints, and the lemma statement follows.
\end{proof}

\loading*
\begin{proof}
	Let $\overline c(z)$ and $\overline c_{\textsc{Fetch}}(z)$ be the classic paging eviction/fetching cost of $z$, i.e. the cost if page fetches/evictions cannot be batched in blocks. The difference between the total fetching cost and eviction cost paid for a single page is at most $c(B(p))$, and hence $\overline c(z) =  \overline c_{\textsc{Fetch}}(z) \pm \sum_{B \in \mathcal{B}} c_B \cdot \bsize$. On the other hand, $\overline c(z) \leq \bsize \cdot c(z)$. Combining these observations:
	\[c_{\textsc{Fetch}}(z) \leq \overline{c}_{\textsc{Fetch}}(z) \leq \overline{c}(z) + \bsize \cdot \sum_{B \in \mathcal{B}} c_B \leq \bsize \cdot \left(c(z) +   \sum_{B \in \mathcal{B}} c_B\right). \qedhere\]
\end{proof}

\subsection{The Natural LP has $\Omega(\beta)$ Integrality Gap} \label{sec:basic_LP_formulation}

Consider the following simple LP formulation, where $\sigma \in \{-1, 1\}$ is a fixed constant. We use $x_p^t$ for the fraction of page $p$ missing from the cache at time $t$.
\begin{align}
	\label{eq:naive_lp}
	\begin{array}{|rl|}
		\hline  & \\
		\displaystyle \min_{\phi,x} & \displaystyle \sum_{B,t} c_B \cdot \phi_B^t \\
		\text{subject to} & \\ & \\
		\forall t \in [T]: & x_{p(t)}^t = 0 \\ & \\
		\forall t \in [T], 
		\forall B \in \mathcal{B}, 
		\forall p \in B: & \phi_B^t \geq \sigma \left(x_p^t - x_p^{t-1}\right) \\  & \\
		\forall t \in [T]: & \sum_{p} x_p^t \geq n-k \\ & \\
		\forall t \in [T],
		\forall B \in \mathcal{B}: & \phi_B^t \in [0,1] \\
		\forall t \in [T],
		\forall p \in [n]: & x_p^t \in [0,1] \\
		& \\ \hline
	\end{array}
\end{align}

If $\sigma = 1$, this is the eviction cost model and $\phi_B^t$ denote the fractional extent to which $B$ is evicted at time $t$; if $\sigma = -1$, this is the fetching cost model and $\phi_B^t$ denote the fractional extent to which $B$ is fetched at time $t$.

Unfortunately, for both fetching and eviction cost models, this LP has an integrality gap of $\Omega(\bsize)$. Consider the following instance in which $n=2\bsize$ pages are divided into two blocks $B_1$ and $B_2$. The cache is of size $k=2\bsize-1$, and is initially empty. The request sequence repeats for several rounds. In each round, it requests first all pages from $B_1$ and then all pages from $B_2$.

The integral algorithm must pay at least $1$ per round, since $2\bsize$ pages are requested and the cache is of size $2\bsize - 1$. On the other hand, the fractional solution begins by loading both blocks to extent $(\bsize-1)/\bsize$. Subsequently, when $B_i$ is requested for $i \in \{1,2\}$, it loads $B_i$ to extent $1$ and $B_j$ (where $j\neq i$) to extent $(\bsize-1)/\bsize$. Hence both the fractional fetching and eviction costs per round are $2/\bsize$. We summarize this observation in the following theorem.

\begin{theorem}
	The simple LP relaxation for block-aware caching with both fetching/eviction cost models has an integrality gap of $\Omega(\bsize)$.
\end{theorem}